\documentclass[11pt,a4paper]{scrartcl}
\usepackage{graphicx}
\usepackage{subfigure}
\usepackage{amsfonts}
\usepackage{amsthm}
\usepackage{amssymb}
\usepackage{amstext}
\usepackage{amsmath}
\usepackage{amsthm}
\usepackage{pstricks}
\usepackage{pspicture}
\usepackage{amscd}
\usepackage{setspace}

 \newtheorem{proposition}{Proposition}

\begin{document}
\title{\bf A note on the proportionality between some consistency indices in the AHP}
\author{
{\bf Matteo Brunelli}
\\
{\normalsize  IAMSR and Turku Centre for Computer Science} \\
{\normalsize \AA bo Akademi University}, {\normalsize Joukahainengatan 3-5A,
FIN-20520 \AA bo, Finland}
\\ {\normalsize e--mail:
\texttt{matteo.brunelli@abo.fi}}
\vspace{0.3cm}\\
{\bf Andrew Critch}
\\
{\normalsize Department of Mathematics} \\
{\normalsize University of California}, {\normalsize Berkeley, CA
94720, United States}
\\ {\normalsize e--mail:
\texttt{critch@math.berkeley.edu}}
\vspace{0.3cm}\\
{\bf Michele Fedrizzi}
\\
{\normalsize  Department of Computer and Management Sciences} \\
{\normalsize University of Trento}, {\normalsize Via Inama 5, I-38122 Trento, Italy}
\\ {\normalsize e--mail:
\texttt{michele.fedrizzi@unitn.it}}
}
\date{September, 2010}

\maketitle \thispagestyle{empty}

\begin{center}
{\bf Abstract }
\end{center}

\abstract{\small \noindent Analyzing the consistency of preferences is an important step in decision making with pairwise comparison matrices, and several indices have been proposed in order to estimate it. In this paper we prove the proportionality between some consistency indices in the framework of the Analytic Hierarchy Process. Knowing such equivalences eliminates redundancy in the consideration of evidence for consistent preferences.}

 \vspace{0.3cm}
 \vspace{0.3cm}


\section{Introduction}
Measuring the inconsistency of an $n \times n$ pairwise comparison matrix --- that is, assigning a numerical value to ``how much" the matrix $\mathbf{A}=(a_{ij})_{n \times n}$ deviates from one indicating consistent preferences --- is an important issue in the Analytic Hierarchy Process (AHP), as well as in other alternative methods of decision-making. 

The oldest and most commonly used measure is the consistency index, $CI$, introduced
by Saaty \cite{saaty1977},
\begin{equation}
CI=\frac{\lambda_{\max}-n}{n-1}, \label{CI}
\end{equation}

\noindent where $\lambda_{\max}$ is the maximum eigenvalue of
$\mathbf{A}$. After Saaty, several other authors proposed different
consistency indices in order to find the most suitable way to
estimate ``how far" $\mathbf{A}$ is from the consistency condition

\begin{equation}
a_{ij}a_{jk}=a_{ik} ~  \forall i,j,k.
\label{consistenza moltiplicativa}
\end{equation}

\noindent Note that Saaty's definition (\ref{CI}) is based on the
fact that, for a positive reciprocal matrix, condition
(\ref{consistenza moltiplicativa}) holds if and only if
$\lambda_{\max}=n$.

Appropriate consistency evaluation of elicited
preferences is seen as important largely because the achievement of a
satisfactory consistency level is viewed as a desirable property. The more consistent are the preferences of a decision maker,
the more likely he/she is a reliable expert, has a deep insight into
the problem, and acts with attention and care with respect to the
problem he/she is facing. Conversely, if judgements are far from
consistency, i.e. they are heavily contradictory, it is likely
that they were given with poor competence and care. Several inconsistency indices have been already proposed in literature to estimate the degree of incoherence of judgements  \cite{Barzilai1998,CavalloD'Apuzzo2009,CavalloD'Apuzzo2010,DuszakKoczkodaj1994,GoldenWang1989,Koczkodaj1993,RamikKorviny2010,RamikPerzina2010,SteinMizzi2007}


If two indices are proportional, it is important to know their proportionality for two reasons.  From an empirical point of view, they should not be considered as contributing independent evidence for the consistency of a subject's preferences.  On the other hand, from a mathematical perspective, their equivalence may be taken to suggest that they represent an important quantity.


\section{Pairwise comparison matrices and consistency indices}
Given a set of alternatives $X=\{ x_{1},\ldots,x_{n} \}~(n \geq 2)$, a
pairwise comparison matrix $\mathbf{A}=(a_{ij})_{n \times n}$ is a
matrix $\mathbf{A} \in [1/9,9]^{n \times n}$ with (i) $a_{ii}=1 \;
\forall i$ and (ii) $a_{ij}a_{ji}=1 \; \forall i,j$ where $a_{ij}$
is a multiplicative estimation of the degree of preference of $x_i$ over $x_j$
\cite{saaty1977}. The comparison scale ranging from 1 to 9 was employed by Saaty based on experimental evidence \cite{miller1956} that an individual cannot simultaneously compare more than $7\pm 2$ objects without being confused. A pairwise comparison matrix is considered {\em consistent} if
and only if the following transitivity condition holds:
\begin{equation}
\label{eq:transitivity}
a_{ik}=a_{ij}a_{jk}~  \forall i,j,k.
\end{equation}
If $\mathbf{A}$ is consistent, then there exists a vector $\mathbf{w}=(w_{1},\ldots,w_{n})$ such that
\begin{equation}
\label{eq:ratio}
a_{ij}=\frac{w_i}{w_j}~  \forall i,j.
\end{equation}
In this case, the vector $\mathbf{w}$ can be
obtained using the geometric mean method:
\begin{equation}
\label{eq:mediageometrica}
w_{i}=\left( \prod_{j=1}^{n}a_{ij} \right)^{\frac{1}{n}}~  \forall i .
\end{equation}

Some other types of matrices have been proposed  in order to
pairwise compare alternatives, and perhaps the second best known
approach, after that of Saaty, is based on {\em reciprocal relations}
\cite{tanino}. Reciprocal relations, which are sometimes also called fuzzy preference relations, can be represented by means of
matrices $\mathbf{R}=(r_{ij})_{n \times n}$ with (i) $r_{ii}=0.5
\; \forall i$ and (ii) $r_{ij}+ r_{ji}=1 \; \forall i,j$ where
$r_{ij}$ is an estimation of the degree of preference given to $x_i$ compared with $x_j$.  Tanino calls a reciprocal relation matrix {\em additively consistent}
if
\begin{equation}
\label{transR}
r_{ij}-r_{ik}-r_{kj}+0.5 = 0 ~ \forall i,j,k.
\end{equation}

Pairwise comparison matrices and reciprocal relations are
theoretically interchangeable representations of preferences, relatable by means of a function $f:[1/9,9]\rightarrow
[0,1]$ defined in \cite{fedrizzi} as follows
\begin{equation}
\label{eq:f}
r_{ij}=f(a_{ij})=\frac{1}{2}(1+\log_{9}a_{ij}),
\end{equation}
and its inverse
\begin{equation}
\label{eq:g}
a_{ij}=f^{-1}(r_{ij})=9^{2 (r_{ij}-0.5)} \, .
\end{equation}
Under this transformation, given $\mathbf{A}=(a_{ij})$ and
$\mathbf{R}=(r_{ij})$, if $r_{ij}=f(a_{ij})~\forall i,j$, then
$\mathbf{A}=(a_{ij})$ and $\mathbf{R}=(r_{ij})$ can be considered to represent the same preference configuration.

Besides Saaty's consistency index (\ref{CI}), several other
consistency indices have been proposed in the literature so far, and in
this short paper we establish the proportionality between two pairs of them. Hence, let us first briefly
recall the definitions of the four consistency indices at issue.

\subsection{The Geometric Consistency Index}
\label{sub:Geometric consistency index} The geometric consistency
index \cite{GCI,crawford85} is based on the deviations of the
entries $a_{ij}$ of $\mathbf{A}$ from the consistent values
$w_i / w_j$, where the weight vector
$\mathbf{w}=(w_1,...,w_n)$ is given by (\ref{eq:mediageometrica}).  It has the following formulation:
\begin{equation}
\label{eq:GCI}
GCI=\frac{2}{(n-1)(n-2)}\sum_{i=1}^{n}\sum_{j>i}^{n} \, \ln^2{e_{ij}}
\end{equation}
with $e_{ij}:=a_{ij} (w_j / w_i )$ being a local estimator of
inconsistency and $\frac{2}{(n-1)(n-2)}$ a normalization factor.

\subsection{The index of Lamata and Pel\'aez}
\label{sub:Index of determinants}
The index of Lamata and Pel\'aez \cite{lamatapelaez,pelaezlamata}, denoted by $CI^{*}$, is based on the property that three alternatives $x_i,x_j,x_k$ are pairwise
compared in a consistent way if and only if the determinant of the
corresponding pairwise comparison matrix of order three

\begin{equation}
\mathbf{A}_{3 \times 3} = \left(%
\begin{array}{ccc}
  1 & a_{ij} & a_{ik} \\
  \frac{1}{a_{ij}} & 1 & a_{jk} \\
  \frac{1}{a_{ik}} & \frac{1}{a_{jk}} & 1 \\
\end{array}%
\right)
\label{matrice3X3}
\end{equation}

\noindent is equal to zero,

\begin{equation}
\label{det} \det (\mathbf{A}_{3 \times
3})=\frac{a_{ik}}{a_{ij}a_{jk}}+\frac{a_{ij}a_{jk}}{a_{ik}}-2 = 0.
\end{equation}
\noindent Based on this property, the authors define the
consistency index $CI^{*}$ of an $n \times n$ pairwise comparison
matrix $\mathbf{A}$ as the mean value of the determinants of all
the $3 \times 3$ submatrices of $\mathbf{A}$.

\subsection{The index $c_3$}
\label{sub:The index $c_3$}
Shiraishi et al. \cite{giapponesi1,giapponesi2,giapponesi3} proposed, as a consistency index of a pairwise comparison matrix, the coefficient $c_3$ of its characteristic polynomial.
\begin{displaymath}
P_{\mathbf{A}}(\lambda)=\lambda^n +c_{1}\lambda^{n-1}+\cdots
+c_{n-1}\lambda+c_{n} \; .
\end{displaymath}
\noindent They proved \cite{giapponesi1} that $c_3(\mathbf{A})\leq 0$ for
every pairwise comparison matrix $\mathbf{A}$, with $c_3(\mathbf{A}) = 0$ if and only if $\mathbf{A}$ is consistent, which justifies its use as a consistency index.

\subsection{The index $\rho$}
\label{sub:The index rho} The index $\rho$ for reciprocal
relations \cite{FedFedMarPer2002,fedrizzigiove} is based on an
index of local consistency associated with the triplet
$(x_i,x_j,x_k)$, that is
\begin{equation}
\label{eq:L} t_{ijk}^2=(r_{ij}-r_{ik}-r_{kj}+0.5)^2.
\end{equation}
which clearly derives from (\ref{transR}). Fedrizzi and Giove \cite{fedrizzigiove} defined a global consistency index as
the mean value of the local consistency indices for all the
possible triplets $(x_i,x_j,x_k)$, obtaining
\begin{equation}
\label{eq:rho}
\rho=\sum_{i<j<k}^{n} (r_{ij}-r_{ik}-r_{kj}+0.5)^2
\bigg/ \binom{n}{3}.
\end{equation}

\section{Results}
In this section we prove that the index $c_3$ is proportional to
$CI^{*}$, and the index $\rho$ is proportional to $GCI$.

\begin{proposition}
\label{proposizione1}
Given a positive reciprocal
matrix $\mathbf{A}_{n \times n}$ with $n \geq 3$, the consistency indices $c_3$
and $CI^*$ satisfy the equality
\begin{equation}
c_3 = - \binom{n}{3} CI^*.
\label{prima equivalenza}
\end{equation}
\end{proposition}

\begin{proof}

\noindent Consistency index $CI^*$ is the mean value of the
determinants of all the $3 \times 3$ submatrices
(\ref{matrice3X3}) of $\mathbf{A}$, and therefore,

\begin{equation} CI^{*}=\sum_{i=1}^n \sum_{j>i}^n
\sum_{k>j}^n
\bigg(\frac{a_{ik}}{a_{ij}a_{jk}}+\frac{a_{ij}a_{jk}}{a_{ik}}-2\bigg)\bigg/
\binom{n}{3}.
\label{svil_det_3X3}
\end{equation}

\noindent Furthermore, since $\mathbf{A}$ is positive and
reciprocal, by expanding $P_{\mathbf{A}}(\lambda)$ (see \cite{giapponesi1}) one obtains
\begin{equation}
c_3=\sum_{i=1}^n \sum_{j>i}^n
\sum_{k>j}^n\bigg(2-\frac{a_{ik}}{a_{ij}a_{jk}}-\frac{a_{ij}a_{jk}}{a_{ik}}
 \bigg) .
\label{svilc3}
\end{equation}
\noindent Then, equality (\ref{prima equivalenza}) follows from
(\ref{svil_det_3X3}) and (\ref{svilc3}).
\end{proof}

If in this case the similarity between the two indices was quite
clear, then the same cannot be said about the next two. For this
reason, if the previous proof was straightforward, the next
involves more computations.

\begin{proposition}
\label{proposizione2}
Given a reciprocal relation
$\mathbf{R}=(r_{ij})_{n \times n}$ , the consistency indices
$\rho$ and $GCI$ satisfy the equality
\begin{equation}
\rho = \frac{3}{4\ln^2(9)}GCI
\label{seconda equivalenza}
\end{equation}
\noindent for every  $n \geq 3$
\end{proposition}

\begin{proof}
For later convenience, letting $q_{ij}=r_{ij}-0.5$ allows us to
write $r_{ij}+r_{ji}=1$ property as $q_{ij}=-q_{ji}$.  Then,
(\ref{eq:g}) becomes $a_{ij}=9^{2q_{ij}}$. Now, write
$t_{ijk}=r_{ij}-r_{ik}-r_{kj}+0.5 = q_{ij}+q_{jk}+q_{ki}$ so that,
from (\ref{eq:rho}), the index $\rho$ can be reformulated (see \cite{fedrizzigiove}) as 

\begin{eqnarray*}
\rho &=& \sum_{ijk}^{n} (r_{ij}-r_{ik}-r_{kj}+0.5)^2
\bigg/ 6 \binom{n}{3} \\
 &=& \sum_{ijk}t^2_{ijk} \bigg/ 6 \binom{n}{3} \; .
\end{eqnarray*}

\noindent Let us rewrite the Geometric Consistency Index
(\ref{eq:GCI}) for reciprocal relations by applying (\ref{eq:f}).
From (\ref{eq:mediageometrica}),
$$\log_{9} w_i = \frac{2}{n} \sum_k q_{ik}$$
and thus, from the definition of local inconsistency $e_{ij} :=
a_{ij}\frac{w_j}{w_i}$ in (\ref{eq:GCI}),
\begin{eqnarray*}
n\log_{9}(e_{ij}) &=& 2nq_{ij}+2\sum_k (q_{jk}-q_{ik}) \\
&=& 2\sum_k (q_{ij}+q_{jk} + q_{ki}) \\
&=& 2 \sum_k t_{ijk}
\end{eqnarray*}
so the Geometric Consistency Index equals

\begin{eqnarray*}
GCI &=& \frac{2}{(n-1)(n-2)}\sum_i\sum_{j>i}\ln^2e_{ij}\\
&=& \frac{1}{(n-1)(n-2)}\sum_{ij}\ln^2e_{ij}\\
&=& \frac{\ln^2(9)}{(n-1)(n-2)}\sum_{ij}\left(\frac{2}{n} \sum_k t_{ijk}\right)^2\\
&=& \frac{4\ln^2(9)}{n^2(n-1)(n-2)} \sum_{ij}\left(\sum_k
t_{ijk}\right)^{2}
\end{eqnarray*}

At this point, the proportionality claim $\rho \propto GCI$ is equivalent to
$$\sum_{ijk}t^2_{ijk} \quad \propto \quad \sum_{ij}\left(\sum_k t_{ijk}\right)^2$$
(where the constant of proportionality could depend on $n$).

First, let us compute the LHS:
\begin{eqnarray*}
t_{ijk}^2 &= q_{ij}^2+q_{jk}^2+q_{ki}^2
+ 2(q_{ij}q_{jk}+q_{jk}q_{ki}+q_{ki}q_{ij})
\end{eqnarray*}
Let $S=\displaystyle\sum_{ij}q_{ij}^2$
and $C=\displaystyle\sum_{ijk}q_{ij}q_{jk}$.  Summing the expansion of $t_{ijk}^2$ one term at a time,
$$\sum_{ijk}q_{ij}^2=\sum_k\sum_{ij}q_{ij}^2=nS$$
and by symmetry,
$$\sum_{ijk}q_{jk}^2=\sum_{ijk}q_{ki}^2=nS.$$
Similarly,
$$\sum_{ijk}q_{ij}q_{jk}=\sum_{ijk}q_{jk}q_{ki}=\sum_{ijk}q_{ki}q_{ij}=C.$$
Hence,
$$\text{LHS}=\sum_{ijk}t^2_{ijk}=nS+nS+nS+2(C+C+C)=3(nS+2C).$$

Next let us compute the RHS, first by rewriting:

$$\text{RHS} = \sum_{ij}\left(\sum_k t_{ijk}\right)^2 = \sum_{ij}\left(\sum_{kl} t_{ijk}t_{ijl}\right) = \sum_{ijkl}t_{ijk}t_{ijl}$$

\begin{eqnarray*}
t_{ijk}t_{ijl} &=& (q_{ij}+q_{jk}+q_{ki})(q_{ij}+q_{jl}+q_{li})\\
\mbox{ \ } &=& q^2_{ij}+q_{ij}q_{jl}+q_{ij}q_{li}
+q_{jk}q_{ij}+q_{jk}q_{jl}+q_{jk}q_{li}
+q_{ki}q_{ij}+q_{ki}q_{jl}+q_{ki}q_{li}  \\
\end{eqnarray*}
The 1st term sums to
$$\sum_{ijkl}q_{ij}^2=\sum_{kl}\sum_{ij}q_{ij}^2=n^2S.$$
The 2nd term sums to
$$\sum_{ijkl}q_{ij}q_{jl}=\sum_k\sum_{ijl}q_{ij}q_{jl}=nC.$$
Similarly, the 3rd, 4th, and 7th terms respectively sum to
$$\sum_{ijkl}q_{li}q_{ij}
=\sum_{ijkl}q_{ij}q_{jk}
=\sum_{ijkl}q_{ki}q_{ij}=nC,$$
whereas the 5th and 9th terms each sum to
$$\sum_{ijkl}-q_{kj}q_{jl} =
\sum_{ijkl}-q_{ki}q_{il} = -nC.$$
The 6th term sums to
$$\sum_{ijkl}q_{jk}q_{li}=\left(\sum_{jk}q_{jk}\right)\left(\sum_{li}q_{li}\right)=(0)(0)=0,$$
and similarly the 8th term sums to $0$.

Hence, the total sum is
\begin{eqnarray*}
\text{RHS} &=& n^2S+nC+nC+nC-nC+0+nC+0-nC \\
&=& n^2S+2nC \\
&=& n(nS+2C)
\end{eqnarray*}
so we obtain the proportionality
$$\frac{\text{LHS}}{\text{RHS}} = \frac{3(nS+2C)}{n(nS+2C)} = \frac{3}{n},$$
and recover
\begin{eqnarray*}
\frac{\rho}{GCI} &=& \frac{\text{LHS}}{\text{RHS}}\cdot \frac{n^2(n-1)(n-2)}{4\ln^2(9)}\cdot \frac{1}{6 \binom{n}{3} } \\
&=& \frac{3n(n-1)(n-2)}{4\ln^2(9)}\cdot \frac{1}{n(n-1)(n-2)} \\
&=& \frac{3}{4\ln^2(9)}.
\end{eqnarray*}
\end{proof}

\noindent Note that the constant of proportionality between $c_3$
and $CI^*$ depends on the number $n$ of alternatives, whereas the
one between $\rho$ and $GCI$ does not. Propositions \ref{proposizione1} and \ref{proposizione2} can also
be represented graphically. We randomly generated a large number
of pairwise comparison matrices (or, equivalently, reciprocal
relations) and associated each of them with a point on the
cartesian plane having as coordinates the corresponding values of the two consistency indices involved in proposition \ref{proposizione1}. As expected, all the points lie on a straight line. The same holds for proposition \ref{proposizione2}.

\section{Conclusions}
When making use of the various indices observed and proven proportional in this paper, we believe it is important that the applied mathematician be aware of their equivalence.  This avoids redundancy in the consideration of evidence for consistent preferences, and allows existing results proven for one index to apply directly to other indices which are proportional to it.

\end{document}